\documentclass[12pt]{article}
\usepackage{amsmath,amssymb,amsthm}
\usepackage{tikz}
\newtheorem{theorem}{Theorem}[section]

\newtheorem{proposition}[theorem]{Proposition}

\theoremstyle{definition}
\newtheorem{example}[theorem]{Example}
\newtheorem{definition}[theorem]{Definition}

\begin{document}
\title{\bf Interval-valued fuzzy  graphs} \normalsize
\author{{\bf Muhammad Akram$^{\bf a}$\ and \ Wieslaw A. Dudek$^{\bf b}$} \\
{\small {\bf a.}  Punjab University College of Information
Technology,
University of the Punjab,}\\
{\small  Old Campus, Lahore-54000, Pakistan.}\\
  {\small E-mail: makrammath@yahoo.com,
  ~~~~m.akram@pucit.edu.pk}\\
 {\small {\bf b.}  Institute of Mathematics and Computer Science,  Wroclaw University
of Technology,}\\
 {\small  Wyb. Wyspianskiego 27, 50-370,Wroclaw, Poland.}\\
  {\small E-mail: dudek@im.pwr.wroc.pl} }
\date{}
 \maketitle
 \hrule
\begin{abstract}
We define the Cartesian product, composition, union and join on
interval-valued fuzzy graphs and investigate some of their properties. We
also introduce the notion of interval-valued fuzzy complete graphs
and present some properties of self complementary and self weak
complementary interval-valued fuzzy complete graphs.
\end{abstract}
 \hrule
 \footnote{Corresponding Author:\\  M. Akram
(makrammath@yahoo.com, m.akram@pucit.edu.pk)}

\section{Introduction}
In 1975, Zadeh \cite{LA1} introduced the notion of  interval-valued fuzzy sets  as an extension of fuzzy sets \cite{LA} in which the values of the membership degrees are intervals of numbers instead of the numbers.
 Interval-valued fuzzy sets provide a more adequate description of uncertainty than traditional fuzzy sets. It is therefore important to use interval-valued fuzzy sets in applications, such as fuzzy control. One of the  computationally most intensive part of fuzzy control is defuzzification \cite{JMM}.
   Since interval-valued fuzzy sets are widely studied and used,  we describe briefly the work of Gorzalczany on approximate reasoning \cite{MB1, MB2}, Roy and Biswas on medical diagnosis \cite{MK}, Turksen on
multivalued logic \cite{IB} and  Mendel on intelligent control \cite{JMM}. \\
The fuzzy graph theory as a generalization of Euler's graph theory
was first introduced by Rosenfeld \cite{RA} in 1975. The fuzzy
relations between fuzzy sets were  first considered by  Rosenfeld
and he developed the structure of fuzzy graphs obtaining analogs
of several graph theoretical concepts. Later,  Bhattacharya
\cite{BP} gave some remarks on fuzzy graphs, and some operations on
fuzzy graphs were introduced by Mordeson and Peng \cite{JN1}. The
complement of a fuzzy graph was defined by Mordeson \cite{JN2} and
further studied by Sunitha and Vijayakumar \cite{MS}. Bhutani and
Rosenfeld introduced the concept of $M$-strong fuzzy graphs in
\cite{KR2} and studied some properties. The concept of   strong
arcs in fuzzy graphs was discussed in \cite{KR1}. Hongmei and
Lianhua gave the definition of interval-valued graph in \cite{JH}.\\
In this paper, we define the operations of Cartesian product,
composition, union  and join on interval-valued fuzzy graphs and
investigate some properties.  We study isomorphism (resp. weak
isomorphism) between interval-valued fuzzy graphs is an
equivalence relation (resp. partial order). We introduce the
notion of interval-valued fuzzy complete graphs and present some
properties of self complementary and self weak complementary
interval-valued fuzzy complete graphs.\\
 The definitions and terminologies that we used in this paper are
standard. For other notations, terminologies and applications, the
readers are referred to \cite{MA08, MA, AA, KT, FH, KP, SMS2, JN11,AM1, AP, LAZ75}.

\section{Preliminaries}

A {\it graph} is an ordered pair $G^*=(V,E),$ where $V$ is  the
set of vertices of $G^*$ and $E$ is the set of edges of $G^*$. Two
vertices $x$ and $y$ in a graph $G^*$ are said to be adjacent in
$G^*$ if $\{x,y\}$ is in an edge of $G^*$. (For simplicity an edge
$\{x,y\}$ will be denoted by $xy$.) A {\it simple graph} is a
graph without  loops and multiple edges. A {\it complete
graph} is a simple graph in which every pair of distinct vertices
is connected by an edge. The complete graph on $n$ vertices has
$n$ vertices and $n(n-1)/2$ edges. We will consider only graphs
with the finite number of vertices and edges.

By a {\it complementary graph} $\overline{G^*}$ of a simple graph
$G^*$ we mean a graph having the same vertices as $G^*$ and such
that two vertices are adjacent in $\overline{G^*}$ if and only if
they are not adjacent in $G^*$.

An {\it isomorphism} of graphs $G^*_1$ and $G^*_2$ is a bijection
between the vertex sets of  $G^*_1$ and $G^*_2$ such that any two
vertices $v_1$ and $v_2$ of $G^*_1$ are adjacent in $G^*_1$ if and
only if $f(v_1)$ and $f(v_2)$ are adjacent in $G^*_2$. Isomorphic
graphs are denoted by $G^*_1 \simeq G^*_2.$

Let $G^*_1=(V_1, E_1)$ and $G^*_2=(V_2, E_2)$ be two simple graphs,
we can construct several new graphs. The first construction called
the {\it Cartesian product} of $G^*_1$ and $G^*_2$ gives a graph
$G^*_1 \times G^*_2=(V, E)$ with $V=V_1 \times V_2$ and
\[
E= \{(x,x_2)(x,y_2)| x\in V_1, x_2y_2\in E_2\}\cup\{(x_1,z)(y_1,
z)|x_1y_1 \in E_1,z\in V_2 \}.
\]
The {\it composition} of graphs $G^*_1$ and $G^*_2$ is the graph
$G^*_1[G^*_2]=(V_1 \times V_2,E^0)$, where
$$
E^0= E\cup\{(x_1,x_2)(y_1,y_2)|x_1y_1 \in E_1, x_2\neq y_2\}
$$
and $E$ is defined as in $G^*_1 \times G^*_2$. Note that
$G^*_1[G^*_2]\neq G^*_2[G^*_1].$

The {\it union} of graphs $G^*_1$ and $G^*_2$ is defined as $G^*_1
\cup G^*_2=(V_1\cup V_2, E_1\cup E_2)$.

The {\it join} of $G^*_1$ and $G^*_2$ is the simple graph $G^*_1 +
G^*_2=(V_1 \cup V_2, E_1 \cup E_2 \cup E')$, where $E'$ is the set
of all edges joining the nodes of $V_1$ and $V_2$. In this
construction it is assumed that $V_1\cap V_2\neq\emptyset$ .

By a {\it fuzzy subset} $\mu$ on a set $X$ is mean a map $\mu
:X\to [0,1]$. A map $\nu: X\times X\to [0,1]$ is called  a {\it
fuzzy relation} on $X$ if $\nu(x,y)\leq \min(\mu(x),\mu(y))$ for
all $x,y\in X$. A fuzzy relation  $\nu$ is {\it symmetric} if
$\nu(x, y)= \nu(y, x)$ for all $x,y\in X$.

An {\it interval number} $D$ is an interval $[a^{-}, a^{+}]$ with
$0\leq a^-\leq a^+\leq 1$. The interval $[a,a]$ is identified with
the number $a\in [0,1]$. $D[0,1]$ denotes the set of all interval
numbers.

For interval numbers $D_1=[a_1^{-}, b_1^{+}]$  and $D_2=[a_2^{-},
b_2^{+}]$, we define
\begin{itemize}
\item ${\rm rmin}(D_1, D_2)={\rm rmin}([a_1^{-}, b_1^{+}],
[a_2^{-}, b_2^{+}])= [\min\{a_1^{-}, a_2^{-}\}, \min\{b_1^{+},
b_2^{+}\}]$,
\item ${\rm rmax}(D_1, D_2)={\rm rmax}([a_1^{-},
b_1^{+}], [a_2^{-}, b_2^{+}])= [\max\{a_1^{-}, a_2^{-}\},
\max\{b_1^{+}, b_2^{+}\}]$,
\item  $D_1 + D_2=[a_1^-+a_2^--a_1^-\cdot a_2^-, b_1^++b_2^+-b_1^+\cdot
b_2^+]$,
\item  $D_1 \leq D_2$  $\Longleftrightarrow$  $a_1^{-} \leq a_2^{-}$ and
$b_1^{+} \leq b_2^{+}$,
\item  $D_1=D_2$     $\Longleftrightarrow$ $a_1^{-} = a_2^{-}$ and $b_1^{+} = b_2^{+}$,
\item  $D_1 <D_2$ $\Longleftrightarrow$  $D_1 \leq D_2$ and $D_1 \neq D_2$,
\item $kD= k[a_1^{-}, b_1^{+}]= [ka_1^{-}, kb_1^{+}]$, where $ 0 \leq  k \leq
1$.
\end{itemize}
Then, $(D[0,1],\leq,\vee,\wedge)$ is a complete lattice with
$[0,0]$ as the least element and $[1,1]$ as the greatest.

The {\it interval-valued fuzzy set} $A$ in $V$ is defined by
\[
A=\{(x, [\mu^-_A(x), \mu^+_A(x)]): x \in V \},
\]
where $\mu^-_A(x)$ and $\mu^+_A(x)$ are fuzzy subsets of $V$ such
that $\mu^-_A (x)\leq \mu^+_A(x)$ for all $x\in V.$ For any two
interval-valued sets $A=[\mu^-_A(x),\mu^+_A(x)]$ and
$B=[\mu^-_B(x), \mu^+_B(x)])$ in $V$ we define:
\begin{itemize}
\item $A\bigcup B=\{(x,\max(\mu^-_A(x),\mu^-_B(x)),\max(\mu^+_A(x),\mu^+_B(x))):x\in V\}$,
\item $A\bigcap B=\{(x,\min(\mu^-_A(x),\mu^-_B(x)),\min(\mu^+_A(x),\mu^+_B(x))):x\in V\}$.
\end{itemize}

If $G^*=(V,E)$ is a graph, then by an {\it interval-valued fuzzy
relation} $B$ on a set $E$ we mean an interval-valued fuzzy set
such that
\[
\mu^-_B(xy)\leq\min(\mu^-_A(x),\mu^-_A(y)),
\]
\[
 \mu^+_B(xy) \leq \min(\mu^+_A(x),\mu^+_A(y))
\]
for all $xy\in E$.

\section{Operations on interval-valued fuzzy graphs}

Throughout in this paper, $G^*$ is a crisp graph, and $G$ is an
interval-valued fuzzy graph.

\begin{definition}
By an {\it interval-valued fuzzy graph} of a graph $G^*=(V,E)$ we
mean a pair $G=(A,B)$, where $A=[\mu^-_A,\mu^+_A]$ is an
interval-valued fuzzy set on $V$ and $B=[\mu^-_B,\mu^+_B]$ is an
interval-valued fuzzy relation on $E$.
\end{definition}

\begin{example}
Consider a graph $G^*=(V, E)$  such that $V=\{x,y,z\}$, $E=\{xy,
yz,zx\}$. Let $A$ be  an  interval-valued fuzzy set of $V$ and let $B$ be an interval-valued fuzzy set of $E \subseteq  V \times V$ defined by
 \[A=< (\frac{x}{0.2}, \frac{y}{0.3}, \frac{z}{0.4}), (\frac{x}{0.4}, \frac{y}{0.5}, \frac{z}{0.5})
 >, \]
 \[B=< (\frac{xy}{0.1}, \frac{yz}{0.2}, \frac{zx}{0.1}), (\frac{xy}{0.3}, \frac{yz}{0.4}, \frac{zx}{0.4})
 >. \]

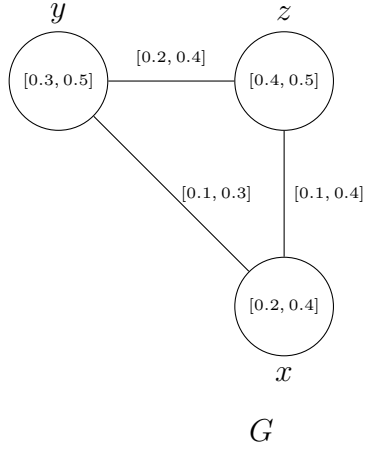
\begin{center}
\begin{tikzpicture}[scale=3]
 \path (1,0.30) node (l) {$y$};
 \path (1.9,-1.55) node (l){$G$};
\path (2,0.30) node (l) {$z$};
\path (2,-1.30) node (l) {$x$};
\path (1.5,0.1) node (l) {\tiny {$[0.2, 0.4]$}};
\path (2.20,-0.5) node (l) {\tiny{$[0.1,0.4]$}};
\path (1.7,-0.5) node (l) {\tiny{$[0.1, 0.3]$}};

\tikzstyle{every node}=[draw,shape=circle];

\path (2,0) node (z) {\tiny{$[0.4, 0.5]$}};
\path (2,-1) node (x) {\tiny{$[0.2, 0.4]$}};
\path (1,0) node (y) {\tiny{$[0.3, 0.5]$}};


\draw
      (y) -- (z)
      (y) -- (x)
       (z) -- (x);
\end{tikzpicture}
\end{center}

By routine computations, it is easy to see that $G=(A, B)$ is an
interval-valued fuzzy graph of $G^*$.
\end{example}

\begin{definition}\label{D-33}
The {\it Cartesian product $G_1\times G_2$ of two interval-valued
fuzzy graphs} $G_1=(A_1,B_1)$ and $G_2=(A_2,B_2)$ of the graphs
$G^*_1=(V_1,E_1)$ and $G^*_2=(V_2,E_2)$ is defined as a pair
$(A_1\times A_1,B_1\times B_2)$ such that
\begin{itemize}
\item [\rm(i)]
$\left\{\begin{array}{ll}(\mu^-_{A_1} \times \mu^-_{A_2})(x_1,
x_2)=\min(\mu^-_{A_1}(x_1),\mu^-_{A_2}(x_2)) \\
(\mu^+_{A_1}\times\mu^+_{A_2})(x_1, x_2)=\min(\mu^+_{A_1}(x_1),
\mu^+_{A_2}(x_2))\end{array}\right.$

for all\ $(x_1, x_2) \in V,$

\item [\rm(ii)]
$\left\{\begin{array}{ll}(\mu^-_{B_1}\times\mu^-_{B_2})((x,x_2)(x,y_2))=\min(\mu^-_{A_1}(x),
\mu^-_{B_2}(x_2y_2))\\
(\mu^+_{B_1} \times\mu^+_{B_2})((x,x_2)(x,y_2))=
\min(\mu^+_{A_1}(x), \mu^+_{B_2}(x_2y_2))\end{array}\right.$

for all $x\in V_1$ and $x_2y_2 \in E_2$,

\item [\rm(iii)]
$\left\{\begin{array}{ll}(\mu^-_{B_1}\times \mu^-_{B_2})((x_1,z)(y_1,z))=\min(\mu^-_{B_1}(x_1y_1),\mu^-_{A_2}(z))\\
(\mu^+_{B_1}\times\mu^+_{B_2})((x_1,z)(y_1,z))=\min(\mu^+_{B_1}(x_1y_1),\mu^+_{A_2}(z))\end{array}\right.$

for all $z\in V_2$ and $x_1y_1 \in E_1$.
\end{itemize}
\end{definition}

\begin{example}\label{Ex-34} Let $G^*_1=(V_1, E_1)$ and $G^*_2=(V_2, E_2)$ be graphs such that
$V_1=\{a, b\}$, $V_2=\{c, d\}$, $E_1=\{ab\}$  and $E_2=\{cd\}$.
Consider two interval-valued fuzzy graphs $G_1=(A_1,B_1)$ and
$G_2=(A_2,B_2)$, where
 \[A_1=< (\frac{a}{0.2}, \frac{b}{0.3}), (\frac{a}{0.4}, \frac{b}{0.5})
 >, \ \ \ \ \ \ B_1=< \frac{ab}{0.1}, \frac{ab}{0.2} >, \]
 \[A_2=< (\frac{c}{0.1}, \frac{d}{0.2}), (\frac{c}{0.4}, \frac{d}{0.6})
 >, \ \ \ \ \ \ B_2=< \frac{cd}{0.1}, \frac{cd}{0.3} >. \]
Then, as it is not difficult to verify
\[  (\mu^-_{B_1}\times\mu^-_{B_2})((a,c)(a,d))=0.1, \ \ \ \ \
\ (\mu^+_{B_1}\times\mu^+_{B_2})((a,c)(a,d))=0.3,\]
\[  (\mu^-_{B_1}\times\mu^-_{B_2})((a,c)(b,c))=0.1, \ \ \ \ \ \ (\mu^+_{B_1}\times\mu^+_{B_2})((a,c)(b,c))=0.2,\]
\[  (\mu^-_{B_1}\times\mu^-_{B_2})((a,d)(b,d))=0.1, \ \ \ \ \ \ (\mu^+_{B_1}\times\mu^+_{B_2})((a,d)(b,d))=0.2,\]
\[  (\mu^-_{B_1}\times\mu^-_{B_2})((b,c)(b,d))=0.1, \ \ \ \ \ \ (\mu^+_{B_1}\times\mu^+_{B_2})((b,c)(b,d))=0.3.\]
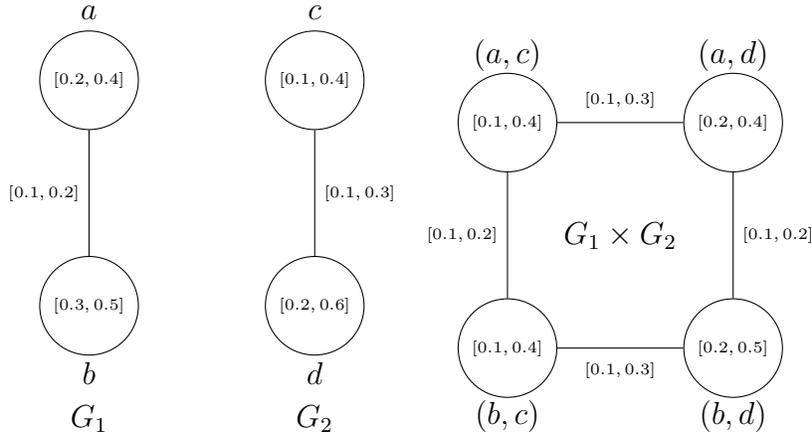
\begin{center}

\begin{tikzpicture}[scale=3]
\path (1,0.30) node (l) {$c$};
\path (1,-1.50) node (l) {$G_2$};
\path (0,0.30) node (l) {$a$};
\path (0,-1.50) node (l) {$G_1$};
\path (0,-1.30) node (l) {$b$};
\path (1,-1.30) node (l) {$d$};
\path (-0.2,-0.5) node (l) {\tiny{$[0.1,0.2]$}};
\path (1.20,-0.5) node (l) {\tiny{$[0.1,0.3]$}};
\tikzstyle{every node}=[draw,shape=circle];

\path (0,0) node (a) {\tiny{$[0.2,0.4]$}};
\path (1,0) node (b) {\tiny{$[0.1,0.4]$}};
\path (0,-1) node (d) {\tiny{$[0.3,0.5]$}};
\path (1,-1) node (e) {\tiny{$[0.2,0.6]$}};


\draw (a) -- (d)
      (b) -- (e);

\end{tikzpicture}
\begin{tikzpicture}[scale=3]

\path (0.5,-0.5) node (l) {$G_1 \times G_2$};
\path (1,0.30) node (l) {$(a,d)$};
\path (0,0.30) node (l) {$(a,c)$};
\path (0,-1.30) node (l) {$(b,c)$};
\path (1,-1.30) node (l) {$(b,d)$};
\path (0.5,0.1) node (l) {\tiny {$[0.1,0.3]$}};

\path (-0.2,-0.5) node (l) {\tiny{$[0.1,0.2]$}};
\path (1.20,-0.5) node (l) {\tiny{$[0.1,0.2]$}};
\path (0.5,-1.10) node (l) {\tiny{$[0.1,0.3]$}};

\tikzstyle{every node}=[draw,shape=circle];

\path (0,0) node (a) {\tiny{$[0.1,0.4]$}};
\path (1,0) node (b) {\tiny{$[0.2,0.4]$}};
\path (0,-1) node (d) {\tiny{$[0.1,0.4]$}};
\path (1,-1) node (e) {\tiny{$[0.2,0.5]$}};


\draw (a) -- (b)
      (a) -- (d)
      (d) -- (e)
      (b) -- (e);
\end{tikzpicture}
\end{center}
By routine computations, it is easy to see that $G_1\times G_2$ is
an interval-valued fuzzy graph of $G^*_1\times G^*_2$.
\end{example}

\begin{proposition}
The Cartesian product $G_1\times G_2=(A_1\times A_2,B_1\times
B_2)$ of two interval-valued fuzzy graphs of the graphs $G^*_1$
and $G^*_2$ is an interval-valued fuzzy graph of $G^*_1 \times
G^*_2$.
\end{proposition}
\begin{proof}
We verify only conditions for $B_1\times B_2$ because conditions
for $A_1\times A_2$ are obvious.

Let $x\in V_1$, $x_2y_2\in E_2$. Then
\begin{eqnarray*}
(\mu^-_{B_1} \times  \mu^-_{B_2})((x, x_2)(x, y_2))   &=& \min(\mu^-_{A_1}(x), \mu^-_{B_2}(x_2y_2)) \\
   &\leq& \min(\mu^-_{A_1}(x), \min(\mu^-_{A_2}(x_2), \mu^-_{A_2}(y_2))) \\
   &=& \min(\min(\mu^-_{A_1}(x), \mu^-_{A_2}(x_2)), \min(\mu^-_{A_1}(x),
   \mu^-_{A_2}(y_2)))\\
   &=& \min((\mu^-_{A_1} \times \mu^-_{A_2})(x, x_2), (\mu^-_{A_1} \times \mu^-_{A_2})(x,
   y_2)),\\[4pt]
(\mu^+_{B_1} \times  \mu^+_{B_2})((x, x_2)(x, y_2))   &=& \min(\mu^+_{A_1}(x), \mu^+_{B_2}(x_2y_2)) \\
   &\leq& \min(\mu^+_{A_1}(x), \min(\mu^+_{A_2}(x_2), \mu^+_{A_2}(y_2))) \\
   &=& \min(\min(\mu^+_{A_1}(x), \mu^+_{A_2}(x_2)), \min(\mu^+_{A_1}(x),
   \mu^+_{A_2}(y_2)))\\
   &=& \min((\mu^+_{A_1} \times \mu^+_{A_2})(x, x_2), (\mu^+_{A_1} \times \mu^+_{A_2})(x,
   y_2)).
\end{eqnarray*}
   Similarly for $z\in V_2$ and $x_1y_1\in E_1$ we have
   \begin{eqnarray*}
(\mu^-_{B_1} \times  \mu^-_{B_2})((x_1, z)(y_1, z))
&=&\min(\mu^-_{B_1}(x_1 y_1),\mu^-_{A_2}(z))  \\
   &\leq& \min( \min(\mu^-_{A_1}(x_1), \mu^-_{A_1}(y_1)), \mu^-_{A_2}(z)) \\
   &=& \min(\min(\mu^-_{A_1}(x), \mu^-_{A_2}(z)), \min(\mu^-_{A_1}(y_1),
   \mu^-_{A_2}(z)))\\
   &=& \min((\mu^-_{A_1} \times \mu^-_{A_2})(x_1, z), (\mu^-_{A_1} \times \mu^-_{A_2})(y_1,
z)),\\[4pt]
(\mu^+_{B_1}\times\mu^+_{B_2})((x_1, z)(y_1, z))
&=&\min(\mu^+_{B_1}(x_1 y_1),\mu^+_{A_2}(z))  \\
   &\leq& \min( \min(\mu^+_{A_1}(x_1), \mu^+_{A_1}(y_1)), \mu^+_{A_2}(z)) \\
   &=& \min(\min(\mu^+_{A_1}(x), \mu^+_{A_2}(z)), \min(\mu^+_{A_1}(y_1),
   \mu^+_{A_2}(z)))\\
   &=& \min((\mu^+_{A_1} \times \mu^+_{A_2})(x_1, z), (\mu^+_{A_1} \times \mu^+_{A_2})(y_1,
z)).
\end{eqnarray*}
This completes the proof.
\end{proof}

\begin{definition}\label{D-36}
The {\it composition $G_1[G_2]=(A_1\circ A_2,B_1\circ B_2)$ of two
interval-valued fuzzy graphs} $G_1$ and $G_2$ of the graphs
$G^*_1$ and $G^*_2$ is defined as follows:
\begin{itemize}
    \item [\rm (i)] $\left\{\begin{array}{ll}(\mu^-_{A_1} \circ \mu^-_{A_2})(x_1, x_2)=\min(\mu^-_{A_1}(x_1),
    \mu^-_{A_2}(x_2)) \\
(\mu^+_{A_1} \circ \mu^+_{A_2})(x_1, x_2)=\min(\mu^+_{A_1}(x_1),
    \mu^+_{A_2}(x_2))\end{array}\right.$

for all $(x_1, x_2) \in V,$
     \item [\rm(ii)]
$\left\{\begin{array}{ll}(\mu^-_{B_1}\circ
\mu^-_{B_2})((x,x_2)(x,y_2))=\min(\mu^-_{A_1}(x),\mu^-_{B_2}(x_2y_2))\\
(\mu^+_{B_1}\circ\mu^+_{B_2})((x,x_2)(x,y_2))=\min(\mu^+_{A_1}(x),
\mu^+_{B_2}(x_2y_2))\end{array}\right.$

for all $x\in V_1$ and $x_2y_2 \in E_2$,

\item [\rm(iii)]
$\left\{\begin{array}{ll}(\mu^-_{B_1} \circ \mu^-_{B_2})((x_1,z)(y_1, z))=\min(\mu^-_{B_1}(x_1y_1), \mu^-_{A_2}(z))\\
(\mu^+_{B_1} \circ \mu^+_{B_2})((x_1,z)(y_1,z))=
\min(\mu^+_{B_1}(x_1y_1), \mu^+_{A_2}(z))\end{array}\right.$

for all $z\in V_2$ and $x_1y_1\in E_1$,

\item [\rm(iv)]  $\left\{\begin{array}{ll}(\mu^-_{B_1}\circ\mu^-_{B_2})((x_1, x_2)(y_1, y_2))=
\min(\mu^-_{A_2}(x_2), \mu^-_{A_2}(y_2),\mu^-_{B_1}(x_1y_1))\\
(\mu^+_{B_1} \circ \mu^+_{B_2})((x_1, x_2) (y_1,y_2))=
\min(\mu^+_{A_2}(x_2), \mu^+_{A_2}(y_2),
\mu^+_{B_1}(x_1y_1))\end{array}\right.$

for all $(x_1, x_2)(y_1, y_2) \in E^0-E$.
\end{itemize}
\end{definition}

\begin{example}\label{Ex-37} Let $G^*_1$ and $G^*_2$ be as in the previous example.
Consider two interval-valued fuzzy graphs $G_1=(A_1,B_1)$ and
$G_2=(A_2,B_2)$ defined by
 \[A_1=< (\frac{a}{0.2}, \frac{b}{0.3}), (\frac{a}{0.5}, \frac{b}{0.5})
 >, \ \ \ \ \ \ B_1=< \frac{ab}{0.2},\frac{ab}{0.4}>, \]
 \[A_2=< (\frac{c}{0.1}, \frac{d}{0.3}), (\frac{c}{0.4}, \frac{d}{0.6})
 >, \ \ \ \ \ \ B_2=<\frac{cd}{0.1}, \frac{cd}{0.3} >. \]
Then we have
\[(\mu^-_{B_1}\circ\mu^-_{B_2})((a,c)(a,d))=0.2,\ \ \ \ \ \ (\mu^+_{B_1}\circ\mu^+_{B_2})((a,c)(a,d))=0.3, \]
\[(\mu^-_{B_1}\circ\mu^-_{B_2})((b,c)(b,d))=0.1,\ \ \ \ \ \ (\mu^+_{B_1}\circ\mu^+_{B_2})((b,c)(b,d))=0.3, \]
\[(\mu^-_{B_1}\circ\mu^-_{B_2})((a,c)(b,c))=0.1,\ \ \ \ \ \ (\mu^+_{B_1}\circ\mu^+_{B_2})((a,c)(b,c))=0.4, \]
\[(\mu^-_{B_1}\circ\mu^-_{B_2})((a,d)(b,d))=0.2,\ \ \ \ \ \ (\mu^+_{B_1}\circ\mu^+_{B_2})((a,d)(b,d))=0.4, \]
\[(\mu^-_{B_1}\circ\mu^-_{B_2})((a,c)(b,d))=0.1,\ \ \ \ \ \ (\mu^+_{B_1}\circ\mu^+_{B_2})((a,c)(b,d))=0.4, \]
\[(\mu^-_{B_1}\circ\mu^-_{B_2})((b,c)(a,d))=0.1,\ \ \ \ \ \ (\mu^+_{B_1}\circ\mu^+_{B_2})((b,c)(a,d))=0.4. \]
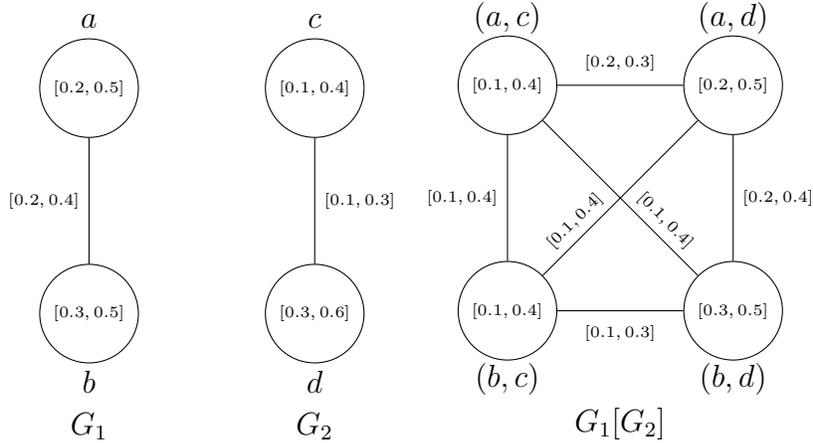
\begin{center}

\begin{tikzpicture}[scale=3]
\path (1,0.30) node (l) {$c$};
\path (1,-1.50) node (l) {$G_2$};
\path (0,0.30) node (l) {$a$};
\path (0,-1.50) node (l) {$G_1$};
\path (0,-1.30) node (l) {$b$};
\path (1,-1.30) node (l) {$d$};
\path (-0.2,-0.5) node (l) {\tiny{$[0.2,0.4]$}};
\path (1.20,-0.5) node (l) {\tiny{$[0.1,0.3]$}};
\tikzstyle{every node}=[draw,shape=circle];

\path (0,0) node (a) {\tiny{$[0.2,0.5]$}};
\path (1,0) node (b) {\tiny{$[0.1,0.4]$}};
\path (0,-1) node (d) {\tiny{$[0.3,0.5]$}};
\path (1,-1) node (e) {\tiny{$[0.3,0.6]$}};


\draw (a) -- (d)
      (b) -- (e);

\end{tikzpicture}
\begin{tikzpicture}[scale=3]
\path (1,0.30) node (l) {$(a,d)$};
\path (0.5,-1.50) node (l) {$G_1 [G_2]$};
\path (0,0.30) node (l) {$(a,c)$};
\path (0,-1.30) node (l) {$(b,c)$};
\path (1,-1.30) node (l) {$(b,d)$};
\path (0.5,0.1) node (l) {\tiny {$[0.2,0.3]$}};

\path (-0.2,-0.5) node (l) {\tiny{$[0.1,0.4]$}};
\path (1.20,-0.5) node (l) {\tiny{$[0.2,0.4]$}};
\path (0.5,-1.10) node (l) {\tiny{$[0.1,0.3]$}};
\path (0.7,-0.6) node (l)[rotate=-45] {\tiny{$[0.1,0.4]$}};
\path (0.30,-0.6) node (l)[rotate=45] {\tiny{$[0.1,0.4]$}};

\tikzstyle{every node}=[draw,shape=circle];

\path (0,0) node (a) {\tiny{$[0.1,0.4]$}};
\path (1,0) node (b) {\tiny{$[0.2,0.5]$}};
\path (0,-1) node (d) {\tiny{$[0.1,0.4]$}};
\path (1,-1) node (e) {\tiny{$[0.3,0.5]$}};


\draw (a) -- (b)

      (a) -- (d)
      (d) -- (e)
      (b) -- (e)
      (a) -- (e)
      (b) -- (d);
\end{tikzpicture}

\end{center}
By routine computations, it is easy to see that $G_1[G_2]= (A_1
\circ A_2, B_1 \circ B_2)$ is an interval-valued fuzzy graph of
$G^*_1[G^*_2]$.
 \end{example}

\begin{proposition}\label{P-38}
The composition $G_1[G_2]$ of interval-valued fuzzy graphs $G_1$
and $G_2$ of $G^*_1$ and $G^*_2$ is an interval-valued fuzzy graph
of $G^*_1[G^*_2]$.
\end{proposition}
\begin{proof}
Similarly as in the previous proof we verify the conditions for
$B_1\circ B_2$ only.

In the case $x\in V_1$, $x_2 y_2\in E_2$, according to $(ii)$ we
obtain
\begin{eqnarray*}
(\mu^-_{B_1} \circ  \mu^-_{B_2})((x, x_2)(x, y_2))   &=& \min(\mu^-_{A_1}(x), \mu^-_{B_2}(x_2y_2))  \\
   &\leq& \min(\mu^-_{A_1}(x), \min(\mu^-_{A_2}(x_2), \mu^-_{A_2}(y_2))) \\
   &=& \min(\min(\mu^-_{A_1}(x), \mu^-_{A_2}(x_2)), \min(\mu^-_{A_1}(x),
   \mu^-_{A_2}(y_2)))\\
   &=& \min((\mu^-_{A_1} \circ \mu^-_{A_2})(x, x_2), (\mu^-_{A_1} \circ \mu^-_{A_2})(x,
   y_2)),\\[4pt]
(\mu^+_{B_1}\circ  \mu^+_{B_2})((x, x_2)(x, y_2))   &=& \min(\mu^+_{A_1}(x), \mu^+_{B_2}(x_2y_2))  \\
   &\leq& \min(\mu^+_{A_1}(x), \min(\mu^+_{A_2}(x_2), \mu^+_{A_2}(y_2))) \\
   &=& \min(\min(\mu^+_{A_1}(x), \mu^+_{A_2}(x_2)), \min(\mu^+_{A_1}(x),
   \mu^+_{A_2}(y_2)))\\
   &=& \min((\mu^+_{A_1}\circ \mu^+_{A_2})(x, x_2), (\mu^+_{A_1} \circ \mu^+_{A_2})(x,
   y_2)).
\end{eqnarray*}

In the case $z\in V_2$, $x_1y_1\in E_1$ the proof is similar.

In the case $(x_1, x_2)(y_1, y_2)\in E^0-E$ we have $x_1y_1\in
E_1$ and $x_2\neq y_2$, which according to $(iv)$ implies
\begin{eqnarray*}
(\mu^-_{B_1}\circ\mu^-_{B_2})((x_1,x_2)(y_1,y_2))
&=&\min(\mu^-_{A_2}(x_2),\mu^-_{A_2}(y_2),\mu^-_{B_1}(x_1y_1)) \\
   &\leq& \min(\mu^-_{A_2}(x_2), \mu^-_{A_2}(y_2), \min(\mu^-_{A_1}(x_1), \mu^-_{A_1}(y_1)) )\\
   &=& \min(\min(\mu^-_{A_1}(x_1), \mu^-_{A_2}(x_2)), \min(\mu^-_{A_1}(y_1),
   \mu^-_{A_2}(y_2)))\\
   &=&\min((\mu^-_{A_1}\circ\mu^-_{A_2})(x_1,
   x_2),(\mu^-_{A_1}\circ\mu^-_{A_2})(y_1,y_2)),\\[4pt]
(\mu^+_{B_1}\circ \mu^+_{B_2})((x_1, x_2)(y_1, y_2))
&=&\min(\mu^+_{A_2}(x_2),\mu^+_{A_2}(y_2),\mu^+_{B_1}(x_1y_1)) \\
   &\leq& \min(\mu^+_{A_2}(x_2), \mu^+_{A_2}(y_2), \min(\mu^+_{A_1}(x_1), \mu^+_{A_1}(y_1)) )\\
   &=& \min(\min(\mu^+_{A_1}(x_1), \mu^+_{A_2}(x_2)), \min(\mu^+_{A_1}(y_1),
   \mu^+_{A_2}(y_2)))\\
   &=& \min((\mu^+_{A_1} \circ \mu^+_{A_2})(x_1, x_2), (\mu^+_{A_1} \circ \mu^+_{A_2})(y_1,
y_2)).
\end{eqnarray*}
This completes the proof.
\end{proof}

\begin{definition}\label{D-39}
The {\it union $G_1 \cup G_2=(A_1 \cup A_2, B_1 \cup B_2)$ of two
interval-valued fuzzy graphs} $G_1$ and $G_2$ of the graphs
$G^*_1$ and $G^*_2$ is defined as follows:

\begin{itemize}
\item[(A)] \
$\left\{\begin{array}{ll}
 (\mu^-_{A_1} \cup\mu^-_{A_2})(x)=\mu^-_{A_1}(x) \ \ \ {\rm if } \
x\in V_1 \ {\rm and } \ x\not\in{V_2}, \\
(\mu^-_{A_1} \cup \mu^-_{A_2})(x)=\mu^-_{A_2}(x) \ \ \ {\rm if } \
x\in V_2 \ {\rm and } \ x\not\in{V_1},\\
(\mu^-_{A_1}\cup\mu^-_{A_2})(x)=\max(\mu^-_{A_1}(x),\mu^-_{A_2}(x))
\ \ {\rm if } \ x\in V_1\cap V_2,
\end{array}\right.$
\item[(B)] \
$\left\{\begin{array}{ll}
 (\mu^+_{A_1} \cup \mu^+_{A_2})(x)=\mu^+_{A_1}(x) \ \ \ {\rm if } \ x\in V_1 \ {\rm and } \ x\not\in{V_2}, \\
(\mu^+_{A_1}\cup\mu^+_{A_2})(x)=\mu^+_{A_2}(x) \ \ \ {\rm if } \
x\in V_2 \ {\rm and } \ x\not\in{V_1},\\
(\mu^+_{A_1}\cup\mu^+_{A_2})(x)=\max(\mu^+_{A_1}(x),\mu^+_{A_2}(x))
\ \ {\rm if } \ x\in V_1\cap V_2,
\end{array}\right.$
\item[(C)] \
$\left\{\begin{array}{ll}
 (\mu^-_{B_1} \cup \mu^-_{B_2})(xy)=\mu^-_{B_1}(xy) \ \ \ {\rm if } \ xy\in E_1 \ {\rm and } \ xy\not\in{E_2}, \\
(\mu^-_{B_1}\cup\mu^-_{B_2})(xy)= \mu^-_{B_2}(xy) \ \ \ {\rm if }
\ xy\in E_2 \ {\rm and } \ xy\not\in{E_1},\\
(\mu^-_{B_1}\cup\mu^-_{B_2})(xy)=\max(\mu^-_{B_1}(xy),\mu^-_{B_2}(xy))
\ \ {\rm if } \ xy\in E_1\cap E_2,
\end{array}\right.$
\item[(D)] \
$\left\{\begin{array}{ll}
 (\mu^+_{B_1} \cup \mu^+_{B_2})(xy)=\mu^+_{B_1}(xy) \ \ \ {\rm if } \ xy\in E_1 \ {\rm and } \ xy\not\in{E_2},\\
(\mu^+_{B_1}\cup\mu^+_{B_2})(xy)=\mu^+_{B_2}(xy) \ \ \ {\rm if } \
xy\in E_2 \ {\rm and } \ xy\not\in{E_1},\\
(\mu^+_{B_1}\cup\mu^+_{B_2})(xy)=\max(\mu^+_{B_1}(xy),\mu^+_{B_2}(xy))
\ \ {\rm if } \ xy\in E_1\cap E_2.
\end{array}\right.$
\end{itemize}
\end{definition}

\begin{example}\label{Ex-310} Let $G^*_1=(V_1, E_1)$ and $G^*_2=(V_2, E_2)$ be graphs such that
$V_1=\{a, b, c, d, e\}$, $E_1=\{ab, bc, be, ce, ad, ed\}$,
$V_2=\{a, b, c, d, f\}$  and $E_2=\{ab, bc, cf, bf, bd\}$.
Consider two interval-valued fuzzy graphs $G_1=(A_1, B_1)$ and
$G_2=(A_2, B_2)$ defined by
\[A_1=< (\frac{a}{0.2}, \frac{b}{0.4}, \frac{c}{0.3}, \frac{d}{0.3}, \frac{e}{0.2}), (\frac{a}{0.4}, \frac{b}{0.5}, \frac{c}{0.6}, \frac{d}{0.7}, \frac{e}{0.6})>,  \]
\[B_1=< (\frac{ab}{0.1}, \frac{bc}{0.2}, \frac{ce}{0.1}, \frac{be}{0.2}, \frac{ad}{0.1}, \frac{de}{0.1}
), (\frac{ab}{0.3}, \frac{bc}{0.4}, \frac{ce}{0.5},
\frac{be}{0.5}, \frac{ad}{0.3}, \frac{de}{0.6})>,\]
\[A_2=< (\frac{a}{0.2}, \frac{b}{0.2}, \frac{c}{0.3}, \frac{d}{0.2}, \frac{f}{0.4}), (\frac{a}{0.4}, \frac{b}{0.5}, \frac{c}{0.6}, \frac{d}{0.6}, \frac{f}{0.6})>,  \]
\[B_2=< (\frac{ab}{0.1}, \frac{bc}{0.2}, \frac{cf}{0.1}, \frac{bf}{0.1}, \frac{bd}{0.2} ),
(\frac{ab}{0.2}, \frac{bc}{0.4}, \frac{cf}{0.5},
\frac{bf}{0.2},\frac{bd}{0.5})>.\] Then, according to the above
definition:

$\begin{array}{ccccc} &(\mu^-_{A_1} \cup \mu^-_{A_2})(a)=0.2,&
(\mu^-_{A_1} \cup
\mu^-_{A_2})(b)=0.4,&\\
&(\mu^-_{A_1} \cup \mu^-_{A_2})(c)=0.3,&
(\mu^-_{A_1} \cup \mu^-_{A_2})(d)=0.3,& \\
&(\mu^-_{A_1}\cup\mu^-_{A_2})(e)=0.2,&(\mu^-_{A_1} \cup
\mu^-_{A_2})(f)=0.4,&\\
&(\mu^+_{A_1} \cup \mu^+_{A_2})(a)=0.4,&(\mu^+_{A_1} \cup
\mu^+_{A_2})(b)=0.5,&\\
(\mu^+_{A_1} \cup \mu^+_{A_2})(c)=0.6,&(\mu^+_{A_1} \cup
\mu^+_{A_2})(d)=0.7,& (\mu^+_{A_1} \cup \mu^+_{A_2})(e)=0.1,&
(\mu^+_{A_1} \cup \mu^+_{A_2})(f)=0.6,\\
(\mu^-_{B_1} \cup \mu^-_{B_2})(ab)=0.1,&(\mu^-_{B_1} \cup
\mu^-_{B_2})(bc)=0.2,&(\mu^-_{B_1} \cup \mu^-_{B_2})(ce)=0.1,&
(\mu^-_{B_1}\cup\mu^-_{B_2})(be)=0.2, \\
(\mu^-_{B_1}\cup\mu^-_{B_2})(ad)=0.1,&(\mu^-_{B_1} \cup
\mu^-_{B_2})(de)=0.1,&(\mu^-_{B_1} \cup \mu^-_{B_2})(bd)=0.2,&
(\mu^-_{B_1}\cup\mu^-_{B_2})(bf)=0.1,\\
(\mu^+_{B_1}\cup \mu^+_{B_3})(ab)=0.3,&(\mu^+_{B_1} \cup
\mu^+_{B_2})(bc)=0.4,&
 (\mu^+_{B_1}\cup\mu^+_{B_2})(ce)=0.5,& (\mu^+_{B_1} \cup
 \mu^+_{B_2})(be)=0.5,\\
(\mu^+_{B_1}\cup \mu^+_{B_2})(ad)=0.3,&(\mu^+_{B_1} \cup
\mu^+_{B_2})(de)=0.6,&(\mu^+_{B_1}\cup \mu^+_{B_2})(bd)=0.5,&
(\mu^+_{B_1}\cup\mu^+_{B_2})(bf)=0.2.\end{array}$
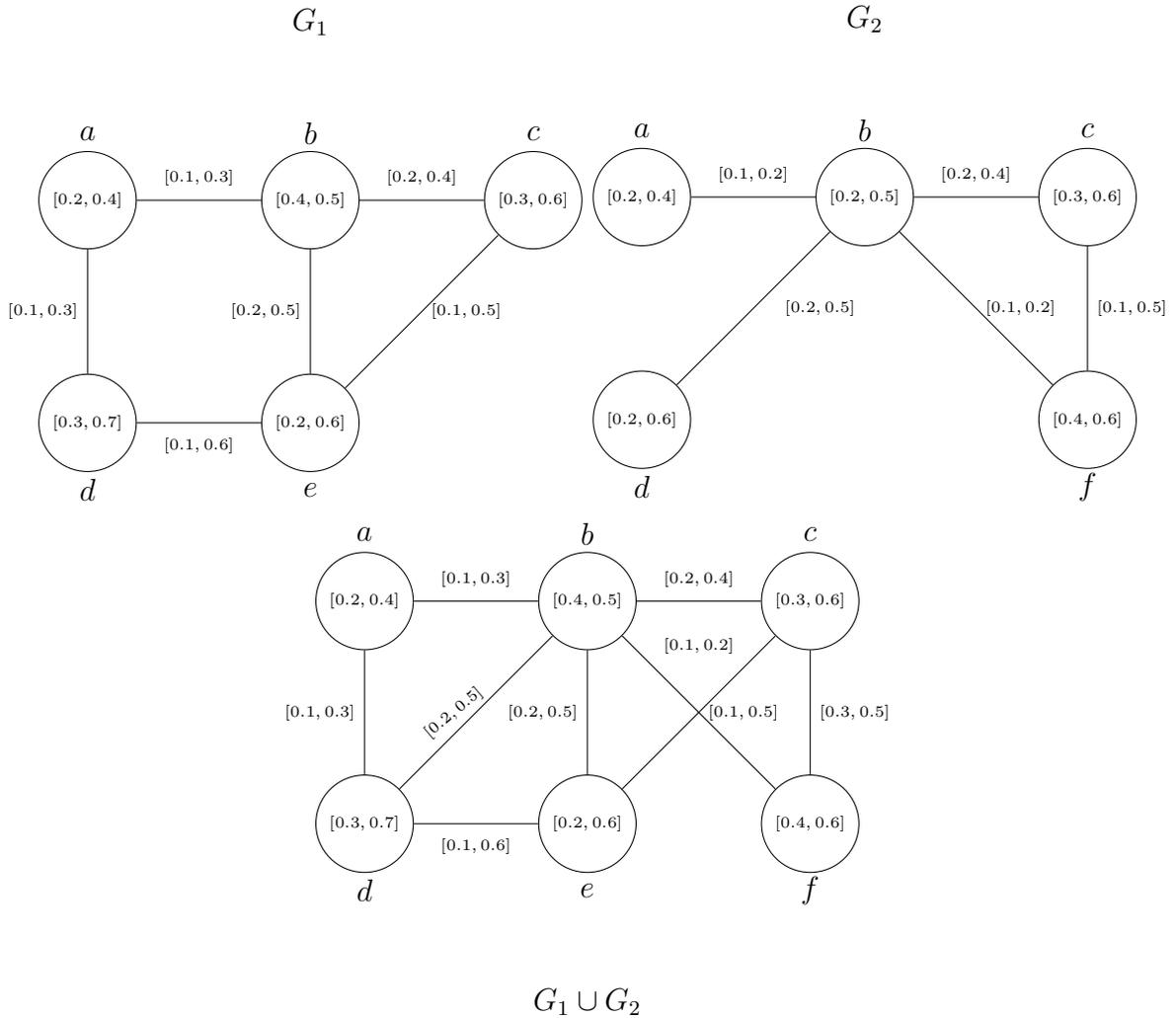
\begin{center}
\begin{tikzpicture}[scale=3]
  \path (1,0.80) node (l){$G_1$};
\path (1,0.30) node (l) {$b$};
\path (0,0.30) node (l) {$a$};
\path (2,0.30) node (l) {$c$};
\path (0,-1.30) node (l) {$d$};
\path (1,-1.30) node (l) {$e$};
\path (0.5,0.1) node (l) {\tiny {$[0.1,0.3]$}};
\path (1.5,0.1) node (l) {\tiny {$[0.2,0.4]$}};
\path (-0.2,-0.5) node (l) {\tiny{$[0.1,0.3]$}};
\path (0.8,-0.5) node (l) {\tiny{$[0.2,0.5]$}};
\path (0.5,-1.10) node (l) {\tiny{$[0.1,0.6]$}};
\path (1.7,-0.5) node (l) {\tiny{$[0.1,0.5]$}};

\tikzstyle{every node}=[draw,shape=circle];

\path (0,0) node (a) {\tiny{$[0.2,0.4]$}};
\path (1,0) node (b) {\tiny{$[0.4,0.5]$}};
\path (2,0) node (c) {\tiny{$[0.3,0.6]$}};
\path (0,-1) node (d) {\tiny{$[0.3,0.7]$}};
\path (1,-1) node (e) {\tiny{$[0.2,0.6]$}};


\draw (a) -- (b)
      (b) -- (c)
      (a) -- (d)
      (d) -- (e)
      (b) -- (e)
      (c) -- (e);
\end{tikzpicture}
\begin{tikzpicture}[scale=3]

\path (1,0.80) node (l) {$G_2$};
\path (1,0.30) node (l) {$b$};
\path (0,0.30) node (l) {$a$};
\path (2,0.30) node (l) {$c$};
\path (0,-1.30) node (l) {$d$};
\path (2,-1.30) node (l) {$f$};
\path (0.5,0.1) node (l) {\tiny {$[0.1,0.2]$}};
\path (1.5,0.1) node (l) {\tiny {$[0.2,0.4]$}};
\path (0.8,-0.5) node (l) {\tiny{$[0.2,0.5]$}};
\path (2.20,-0.5) node (l) {\tiny{$[0.1,0.5]$}};
\path (1.7,-0.5) node (l) {\tiny{$[0.1,0.2]$}};

\tikzstyle{every node}=[draw,shape=circle];

\path (0,0) node (a) {\tiny{$[0.2,0.4]$}};
\path (1,0) node (b) {\tiny{$[0.2,0.5]$}};
\path (2,0) node (c) {\tiny{$[0.3,0.6]$}};
\path (0,-1) node (d) {\tiny{$[0.2,0.6]$}};
\path (2,-1) node (f) {\tiny{$[0.4,0.6]$}};


\draw (a) -- (b)
      (b) -- (c)
      (b) -- (f)
      (b) -- (d)
      (c) -- (f);
\end{tikzpicture}
\begin{tikzpicture}[scale=3]

\path (1,-1.80) node (l) {$G_1 \cup G_2$};
\path (1,0.30) node (l) {$b$};
\path (0,0.30) node (l) {$a$};
\path (2,0.30) node (l) {$c$};
\path (0,-1.30) node (l) {$d$};
\path (1,-1.30) node (l) {$e$};
\path (2,-1.30) node (l) {$f$};
\path (0.5,0.1) node (l) {\tiny {$[0.1,0.3]$}};
\path (1.5,0.1) node (l) {\tiny {$[0.2,0.4]$}};
\path (-0.2,-0.5) node (l) {\tiny{$[0.1,0.3]$}};
\path (0.8,-0.5) node (l)  {\tiny{$[0.2,0.5]$}};
\path (0.5,-1.10) node (l) {\tiny{$[0.1,0.6]$}};
\path (1.7,-0.5) node (l) {\tiny{$[0.1,0.5]$}};
\path (2.20,-0.5) node (l) {\tiny{$[0.3,0.5]$}};
\path (0.4,-0.5) node (l)[rotate=40] {\tiny{$[0.2,0.5]$}};
\path (1.5,-0.2) node (l) {\tiny{$[0.1,0.2]$}};
\tikzstyle{every node}=[draw,shape=circle];

\path (0,0) node (a) {\tiny{$[0.2,0.4]$}};
\path (1,0) node (b) {\tiny{$[0.4,0.5]$}};
\path (2,0) node (c) {\tiny{$[0.3,0.6]$}};
\path (0,-1) node (d) {\tiny{$[0.3,0.7]$}};
\path (1,-1) node (e) {\tiny{$[0.2,0.6]$}};
\path (2,-1) node (f) {\tiny{$[0.4,0.6]$}};

\draw (a) -- (b)
      (b) -- (c)
      (a) -- (d)
      (d) -- (e)
      (b) -- (e)
      (c) -- (e)
      (c) -- (f)
      (b) -- (d)
      (b) -- (f);
\end{tikzpicture}
\end{center}
Clearly, $G_1 \cup G_2=(A_1\cup A_2,B_1\cup B_2)$ is an
interval-valued fuzzy graph of the graph $G_1^*\cup G_2^*$.
\end{example}

\begin{proposition}\label{P-311}
The union of two interval-valued fuzzy graphs is an
interval-valued fuzzy graph.
\end{proposition}
\begin{proof}
Let $G_1=(A_1, B_1)$ and $G_2=(A_2, B_2)$ be interval-valued fuzzy
graphs of $G^*_1$ and $G^*_2$, respectively. We prove that
$G_1\cup G_2=(A_1 \cup A_2, B_1 \cup B_2)$ is an interval-valued
fuzzy graph of the graph $G_1^*\cup G_2^*$. Since all conditions
for $A_1\cup A_2$ are automatically satisfied we verify only
conditions for $B_1\cup B_2$.

At first we consider the case when $xy\in E_1\cap E_2$. Then
\begin{eqnarray*}
  (\mu^-_{B_1} \cup \mu^-_{B_2} )(xy) &=& \max(\mu^-_{B_1}(xy), \mu^-_{B_2}(xy)) \\
   &\leq&  \max( \min(\mu^-_{A_1}(x), \mu^-_{A_1}(y)), \min(\mu^-_{A_2}(x), \mu^-_{A_2}(y))) \\
  &=&  \min( \max(\mu^-_{A_1}(x), \mu^-_{A_2}(x)), \max(\mu^-_{A_1}(y), \mu^-_{A_2}(y))) \\
  &=&  \min((\mu^-_{A_1} \cup \mu^-_{A_2})(x), (\mu^-_{A-1} \cup
  \mu^-_{A_2})(y)),\\
  (\mu^+_{B_1} \cup \mu^+_{B_2} )(xy) &=& \max(\mu^+_{B_1}(xy), \mu^+_{B_2}(xy)) \\
   &\leq&  \max( \min(\mu^+_{A_1}(x), \mu^+_{A_1}(y)), \min(\mu^+_{A_2}(x), \mu^+_{A_2}(y))) \\
  &=&  \min( \max(\mu^+_{A_1}(x), \mu^+_{A_2}(x)), \max(\mu^+_{A_1}(y), \mu^+_{A_2}(y))) \\
  &=&  \min((\mu^+_{A_1} \cup \mu^+_{A_2})(x), (\mu^+_{A-1} \cup \mu^+_{A_2})(y)).
\end{eqnarray*}

If $xy\in E_1$ and $xy\not\in{E_2}$, then
\[ (\mu^-_{B_1} \cup \mu^-_{B_2} )(xy)\leq \min((\mu^-_{A_1} \cup \mu^-_{A_2})(x),
(\mu^-_{A_1} \cup \mu^-_{A_2})(y)), \]
 \[ (\mu^+_{B_1} \cup \mu^+_{B_2} )(xy)\leq \min((\mu^+_{A_1} \cup \mu^+_{A_2})(x),
(\mu^+_{A_1} \cup \mu^+_{A_2})(y)). \]

If $xy\in E_2$ and $xy\not\in{E_1}$, then
\[ (\mu^-_{B_1} \cup \mu^-_{B_2} )(xy)\leq \min((\mu^-_{A_1} \cup \mu^-_{A_2})(x),
 (\mu^-_{A_1} \cup \mu^-_{A_2})(y)), \]
\[ (\mu^+_{B_1} \cup \mu^+_{B_2} )(xy)\leq \min((\mu^+_{A_1} \cup \mu^+_{A_2})(x),
 (\mu^+_{A_1} \cup \mu^+_{A_2})(y)). \]

This completes the proof.
\end{proof}

\begin{definition}\label{D-312}
The {\it join $G_1 + G_2=(A_1 + A_2, B_1 + B_2)$ of two
interval-valued fuzzy graphs} $G_1$ and $G_2$ of the graphs
$G^*_1$ and $G^*_2$ is defined as follows:
\begin{itemize}
\item [(A)] \
$\left\{\begin{array}{ll}
(\mu^-_{A_1}+\mu^-_{A_2})(x)=(\mu^-_{A_1}\cup\mu^-_{A_2})(x)\\
(\mu^+_{A_1} + \mu^+_{A_2})(x)=(\mu^+_{A_1}\cup\mu^+_{A_2})(x)
\end{array}\right.$ \ \ \

if $x\in V_1\cup V_2$,
\item[(B)] \
$\left\{\begin{array}{ll}
(\mu^-_{B_1}+\mu^-_{B_2})(xy)=(\mu^-_{B_1}\cup\mu^-_{B_2})(xy)\\
(\mu^+_{B_1} + \mu^+_{B_2})(xy)=(\mu^+_{B_1} \cup\mu^+_{B_2})(xy)
\end{array}\right.$ \ \ \

if $xy \in E_1 \cap E_2$,
\item [(C)] \
$\left\{\begin{array}{ll}
(\mu^-_{B_1}+\mu^-_{B_2})(xy)=\min(\mu^-_{A_1}(x), \mu^-_{A_2}(y))\\
(\mu^+_{B_1} + \mu^+_{B_2})(xy)= \min(\mu^+_{A_1}(x),\mu^+_{A_2}(y))
\end{array}\right.$ \ \ \

if $xy\in E'$, where $E'$ is the set of all edges joining the
nodes of $V_1$ and $V_2$.
\end{itemize}
\end{definition}

\begin{proposition}\label{P-313}
The join of interval-valued fuzzy graphs is an interval-valued
fuzzy graph.
\end{proposition}
\begin{proof}
Let $G_1=(A_1, B_1)$ and $G_2=(A_2, B_2)$ be interval-valued fuzzy
graphs of $G^*_1$ and $G^*_2$, respectively. We prove that
$G_1+G_2=(A_1+A_2, B_1+B_2)$ is an interval-valued fuzzy graph of
the graph $G_1^*+G_2^*$. In view of Proposition \ref{P-311} is
sufficient to verify the case when $xy \in E'$. In this case we
have
\begin{eqnarray*}
  (\mu^-_{B_1} + \mu^-_{B_2} )(xy) &=& \min(\mu^-_{A_1}(x), \mu^-_{A_2}(y)) \\
   &\leq&  \min( (\mu^-_{A_1} \cup \mu^-_{A_2})(x),(\mu^-_{A_1} \cup \mu^-_{A_2})(y)) \\
  &=&   \min((\mu^-_{A_1} + \mu^-_{A_2})(x), (\mu^-_{A_1} +
  \mu^-_{A_2})(y)),\\[4pt]
(\mu^+_{B_1} + \mu^+_{B_2} )(xy) &=& \min(\mu^+_{A_1}(x), \mu^+_{A_2}(y)) \\
   &\leq&  \min( (\mu^+_{A_1} \cup \mu^+_{A_2})(x),(\mu^+_{A_1} \cup \mu^+_{A_2})(y)) \\
  &=&   \min((\mu^+_{A_1} + \mu^+_{A_2})(x), (\mu^+_{A_1} + \mu^+_{A_2})(y)).
\end{eqnarray*}
This completes the proof.
\end{proof}

\begin{proposition}\label{P-314} Let $G^*_1=(V_1, E_1)$ and $G^*_2=(V_2, E_2)$
be crisp graphs with $V_1 \cap V_2 =\emptyset$. Let $A_1$,
$A_2$, $B_1$ and $B_2$ be interval-valued fuzzy subsets of $V_1$,
$V_2$, $E_1$ and $E_2$, respectively. Then  $G_1 \cup G_2=(A_1
\cup A_2, B_1 \cup B_2)$ is an interval-valued fuzzy graph of
$G_1^*\cup G_2^*$ if and only if $G_1=(A_1, B_1)$ and $G_2=(A_2,
B_2)$ are interval-valued fuzzy graphs of $G^*_1$ and $G^*_2$,
respectively.
\end{proposition}
\begin{proof}
Suppose that $G_1 \cup G_2=(A_1 \cup A_2, B_1 \cup B_2)$ is an
interval-valued fuzzy graph of $G_1^*\cup G_2^*$. Let $xy \in
E_1$. Then $xy \notin E_2$ and $x,y\in V_1-V_2$. Thus
\begin{eqnarray*}
\mu^-_{B_1}(xy)   &=& (\mu^-_{B_1} \cup \mu^-_{B_2})(xy)  \\
 &\leq& \min((\mu^-_{A_1} \cup \mu^-_{A_2})(x), (\mu^-_{A_1} \cup
\mu^-_{A_2})(y)) \\
   &=& \min(\mu^-_{A_1}(x), \mu^-_{A_1}(y)),\\
  \mu^+_{B_1}(xy)   &=& (\mu^+_{B_1} \cup \mu^+_{B_2})(xy)  \\
 &\leq& \min((\mu^+_{A_1} \cup \mu^+_{A_2})(x), (\mu^+_{A_1} \cup
\mu^+_{A_2})(y)) \\
   &=& \min(\mu^+_{A_1}(x), \mu^+_{A_1}(y)).
\end{eqnarray*}
This shows that $G_1=(A_1, B_1)$ is an interval-valued fuzzy
graph. Similarly, we can show that $G_2=(A_2, B_2)$ is an
interval-valued fuzzy graph.

The converse statement is given by Proposition \ref{P-311}.
\end{proof}

As a consequence of Propositions \ref{P-313} and \ref{P-314} we
obtain
\begin{proposition}
Let $G^*_1=(V_1, E_1)$ and $G^*_2=(V_2, E_2)$ be crisp graphs and
let  $V_1 \cap V_2 =\emptyset$. Let $A_1$, $A_2$, $B_1$ and
$B_2$ be interval-valued fuzzy subsets of $V_1$, $V_2$, $E_1$ and
$E_2$, respectively. Then  $G_1 + G_2=(A_1 + A_2, B_1 + B_2)$ is
an interval-valued fuzzy graph of $G^*$  if and only if $G_1=(A_1,
B_1)$ and $G_2=(A_2, B_2)$ are interval-valued fuzzy graphs of
$G^*_1$ and $G^*_2$, respectively.
\end{proposition}

\section{Isomorphisms of interval-valued fuzzy graphs}

In this section we characterize various types of (weak)
isomorphisms of interval valued graphs.
\begin{definition}
Let $G_1=(A_1, B_1)$ and $G_2=(A_2, B_2)$ be two interval-valued
fuzzy graphs. A {\it homomorphism} $f:G_1 \to G_2$ is a mapping
$f:V_1 \to V_2$ such that
\begin{itemize}
\item [\rm (a)] \ $\mu^-_{A_1}(x_1)\leq \mu^-_{A_2}(f(x_1))$, \ \ \ $\mu^+_{A_1}(x_1)\leq \mu^+_{A_2}(f(x_1))$,
\item [\rm (b)] \ $\mu^-_{B_1}(x_1y_1)\leq \mu^-_{B_2}(f(x_1)f(y_1))$, \ \ \ $\mu^+_{B_1}(x_1y_1)\leq \mu^+_{B_2}(f(x_1)f(y_1))$
\end{itemize}
for all $x_1 \in V_1$, $x_1y_1 \in E_1$.
\end{definition}

A bijective homomorphism with the property
\begin{itemize}
\item [\rm (c)] \ $\mu^-_{A_1}(x_1)= \mu^-_{A_2}(f(x_1))$, \ \ \ $\mu^+_{A_1}(x_1)=\mu^+_{A_2}(f(x_1))$,
\end{itemize}
is called a {\it weak isomorphism}. A weak isomorphism preserves
the weights of the nodes but not necessarily the weights of the
arcs.

A bijective homomorphism preserving the weights of the arcs but
not necessarily the weights of nodes, i.e., a bijective
homomorphism $f:G_1 \to G_2$ such that
\begin{itemize}
        \item [\rm (d)] $\mu^-_{B_1}(x_1y_1)= \mu^-_{B_2}(f(x_1)f(y_1))$, $\mu^+_{B_1}(x_1y_1) = \mu^+_{B_2}(f(x_1)f(y_1))$
\end{itemize}
for all $x_1y_1 \in V_1$ is called a {\it weak co-isomorphism}.

A bijective mapping $f:G_1 \to G_2$ satisfying $(c)$ and $(d)$ is
called an {\it isomorphism}.

\begin{example}
Consider  graphs $G^*_1=(V_1, E_1)$ and $G^*_2=(V_2, E_2)$ such
that $V_1=\{a_1, b_1\}$, $V_2=\{a_2, b_2\}$, $E_1= \{a_1 b_1\}$
and $E_2= \{a_2 b_2\}$. Let $A_1$, $A_2$, $B_1$ and $B_2$ be
interval-valued fuzzy subsets defined by
\[
A_1=< (\frac{a_1}{0.2}, \frac{b_1}{0.3}), (\frac{a_1}{0.5},
\frac{b_1}{0.6}) >, \ \ \ B_1=<\frac{a_1b_1}{0.1},
\frac{a_1b_1}{0.3}>,
\]
\[A_2=<(\frac{a_2}{0.3}, \frac{b_2}{0.2}), (\frac{a_2}{0.6}, \frac{b_2}{0.5}) >,
\ \ \  B_2=<\frac{a_2b_2}{0.1},\frac{a_2b_2}{0.4}>.
\]
Then, as it is easy to see, $G_1=(A_1, B_1)$ and $G_2=(A_2, B_2)$
are interval-valued fuzzy graphs of $G^*_1$ and $G^*_2$,
respectively. The map $f:V_1\to V_2$ defined by $f(a_1)=b_2$ and
$f(b_1)=a_2$ is a weak isomorphism but it is not an isomorphism.
\end{example}

\begin{example}
Let $G^*_1$ and $G^*_2$ be as in the previous example and let
$A_1$, $A_2$, $B_1$ and $B_2$ be interval-valued fuzzy subsets
defined by
\[
A_1=< (\frac{a_1}{0.2}, \frac{b_1}{0.3}), (\frac{a_1}{0.4},
\frac{b_1}{0.5}) >, \ \ \ B_1=<\frac{a_1b_1}{0.1},
\frac{a_1b_1}{0.3}>,
\]
\[A_2=<(\frac{a_2}{0.4}, \frac{b_2}{0.3}), (\frac{a_2}{0.5}, \frac{b_2}{0.6}) >,
\ \ \  B_2=<\frac{a_2b_2}{0.1}, \frac{a_2b_2}{0.3}>.
\]
Then $G_1=(A_1, B_1)$ and $G_2=(A_2, B_2)$ are interval-valued
fuzzy graphs of $G^*_1$ and $G^*_2$, respectively. The map
$f:V_1\to V_2$ defined by $f(a_1)=b_2$ and $f(b_1)=a_2$ is a weak
co-isomorphism but it is not an isomorphism.
\end{example}

\begin{proposition}
An isomorphism between interval-valued fuzzy graphs is an
equivalence relation.
\end{proposition}

\noindent{\bf Problem.}  Prove or disprove that weak isomorphism
(co-isomorphism) between interval-valued fuzzy graphs is a partial
ordering relation.

\section{Interval-valued fuzzy complete graphs}

\begin{definition}
An interval-valued fuzzy graph $G=(A, B)$ is called {\it complete}
if
\[\mu^-_B(xy)= \min(\mu^-_A(x), \mu^-_A(y)) ~~~ {\rm
and}~~~\mu^+_B(xy)= \min(\mu^+_A(x), \mu^+_A(y))~~~{\rm for~
all}~~ xy\in E.
\]
 \end{definition}

\begin{example}
Consider a graph $G^*=(V, E)$  such that $V=\{x, y, z\}$, $E=
\{xy, yz, zx\}$. If $A$ and $B$ are interval-valued fuzzy subset
defined by
 \[A=< (\frac{x}{0.2}, \frac{y}{0.3}, \frac{z}{0.4}), (\frac{x}{0.4}, \frac{y}{0.5}, \frac{z}{0.5})
 >, \]
 \[B=< (\frac{xy}{0.2}, \frac{yz}{0.3}, \frac{zx}{0.2}), (\frac{xy}{0.4}, \frac{yz}{0.5}, \frac{zx}{0.4})
 >, \]
then $G=(A,B)$ is an interval-valued fuzzy complete graph of
$G^*$.
\end{example}

As a consequence of Proposition \ref{P-38} we obtain
\begin{proposition} If $G=(A,B)$ be an interval-valued fuzzy complete graph, then
also $G[G]$ is an interval-valued fuzzy complete graph.
 \end{proposition}

\begin{definition}
The {\it complement} of an interval-valued fuzzy complete graph
$G=(A,B)$ of $G^*=(V,E)$ is an interval-valued fuzzy complete
graph $\overline{G}=(\overline{A},\overline{B})$ on
$\overline{G^*}=(V,\overline{E})$, where
$\overline{A}=A=[\mu^-_A,\mu^+_A]$ and
$\overline{B}=[\overline{\mu^-}_B, \overline{\mu^+}_B]$ is defined
by
 $$
\begin{aligned}\overline{\mu^-_B}(x y) & = \begin{cases}
    0 & \hbox{if \ $\mu^-_B(x y) >0$,} \\
    \min( \mu^-_A(x), \mu^-_A(y)) & \hbox{if \ if $\mu^-_B(x y)=0$, }\\
\end{cases}
 \end{aligned}$$
$$ \begin{aligned}\overline{\mu^+_B}(x y) & = \begin{cases}
    0 & \hbox{if \ $\mu^+_B(x y) >0$,} \\
    \min( \mu^+_A(x), \mu^+_A(y)) & \hbox{if \ if $\mu^+_B(x y)=0$. }\\
\end{cases}
 \end{aligned}$$
\end{definition}
\begin{definition}
An interval-valued fuzzy complete graph $G=(A, B)$ is called {\it
self complementary} if  $\overline{\overline{G}}=G$.
\end{definition}

\begin{example}
Consider a graph $G^*=(V, E)$  such that $V=\{a, b, c\}$, $E=
\{ab, bc\}$. Then an interval-valued fuzzy graph $G=(A,B)$, where
\[
A=<(\frac{a}{0.1},\frac{b}{0.2},\frac{c}{0.3}),(\frac{a}{0.3},\frac{b}{0.4},
\frac{c}{0.5})>,
\]
\[
B=<(\frac{ab}{0.1},\frac{bc}{0.2}),(\frac{ab}{0.3},\frac{bc}{0.4})>,
\]
is self complementary.
\end{example}

\begin{proposition}
In a self complementary interval-valued fuzzy complete graph
$G=(A,B)$ we have

\medskip$a)$  \ $\sum\limits_{x\neq y} \mu^-_B(x
y)=\sum\limits_{x\neq y} \min( \mu^-_A(x), \mu^-_A(y))$,

$b)$ \ $\sum\limits_{x\neq y} \mu^+_B(x y)=\sum\limits_{x\neq y}
\min( \mu^+_A(x), \mu^+_A(y))$.
\end{proposition}
\begin{proof}
Let $G=(A,B)$ be a self complementary interval-valued fuzzy
complete graph. Then there exists an automorphism $f:V \to V$ such
that $\mu^-_A(f(x))=\mu^-_A(x)$, $\mu^+_A(f(x))=\mu^+_A(x)$,
$\overline{\mu^-_B}(f(x)f(y))=\mu^-_B(xy)$ and
$\overline{\mu^+_B}(f(x)f(y))=\mu^+_B(xy)$ for all $x,y\in V$.
Hence, for $x,y\in V$ we obtain
\[
\mu^-_B(x
y)=\overline{\mu^-}_B(f(x)f(y))=\min(\mu^-_A(f(x)),\mu^-_A(f(y)))=\min(\mu^-_A(x),
\mu^-_A(y)),
\]
which implies $a)$. The proof of $b)$ is analogous.
\end{proof}
\begin{proposition}
  Let $G=(A, B)$ be an   interval-valued fuzzy complete graph.  If $ \mu^-_B(x
y)=  \min( \mu^-_A(x), \mu^-_A(y))$ and $ \mu^+_B(x y)=  \min(
\mu^+_A(x), \mu^+_A(y))$ for all $x$, $y$ $\in V$, then $G$ is
self complementary.
\end{proposition}
\begin{proof}
Let  $G=(A, B)$ be an   interval-valued fuzzy complete graph such
that  $ \mu^-_B(x y)=  \min( \mu^-_A(x), \mu^-_A(y))$ and $
\mu^+_B(x y)=  \min( \mu^+_A(x), \mu^+_A(y))$ for all $x$, $y$
$\in V$. Then
  $G=\overline{G}$ under the identity map $I: V \to V$. So
  $\overline{\overline{G}}=G$. Hence $G$ is  self complementary.
 \end{proof}
 \begin{proposition}
 Let $G_1=(A_1, B_1)$ and $G_2=(A_2, B_2)$ be    interval-valued
fuzzy complete graphs. Then $G_1 \cong G_2$  if and only if
$\overline{G}_1 \cong \overline{G}_2$.
\end{proposition}
\begin{proof}
 Assume that $G_1$ and $G_2$ are isomorphic, there exists a
 bijective map $f: V_1 \to V_2$ satisfying
 \[ \mu^-_{A_1}(x)= \mu^-_{A_2}(f(x)), ~  \mu^+_{A_1}(x)= \mu^+_{A_2}(f(x)) ~~{\rm for ~ all } ~x \in V_1,\]
\[ \mu^-_{B_1}(x y)= \mu^-_{B_2}(f(x)f(y)), ~  \mu^+_{B_1}(xy)= \mu^+_{B_2}(f(x)f(y)) ~~{\rm for ~ all } ~xy \in E_1.\]
By definition of complement, we have
\[ \overline{\mu^-}_{B_1}(xy)=\min(\mu^-_{A_1}(x), \mu^-_{A_1}(y)=\min(\mu^-_{A_2}(f(x)), \mu^-_{A_2}(f(y)) )=\overline{\mu^-}_{B_2}(f(x)f(y)), \]
\[ \overline{\mu^+}_{B_1}(xy)=\min(\mu^+_{A_1}(x), \mu^+_{A_1}(y)=\min(\mu^+_{A_2}(f(x)), \mu^+_{A_2}(f(y)) )=\overline{\mu^+}_{B_2}(f(x)f(y)) ~~{\rm for ~ all} ~xy \in E_1. \]
Hence $\overline{G}_1\cong \overline{G}_2$.

The proof of converse part is straightforward.
\end{proof}

\section{Conclusions}
It is well known that   interval-valued fuzzy sets
constitute a generalization of the notion of  fuzzy sets. The  interval-valued fuzzy  models  give more precision, flexibility and compatibility to the system as compared to the classical and fuzzy  models. So we have
introduced interval-valued fuzzy  graphs  and have presented
several properties in this paper.  The further  study of interval-valued fuzzy  graphs may also  be extended with the following projects:
\begin{itemize}
  \item  an application of interval-valued fuzzy graphs in database theory
   \item  an application of interval-valued fuzzy graphs in an expert system
    \item  an application of interval-valued fuzzy graphs in neural networks
  \item  an interval-valued fuzzy graph method for finding the
shortest paths in networks

\end{itemize}
\noindent{\bf\large Acknowledgement.} The authors are  thankful to the  referees  for  their  valuable comments and suggestions.

\end{document}